\documentclass[a4paper,USenglish,11pt]{article}
\usepackage{a4wide}
\usepackage[utf8]{inputenc}
\usepackage{amsmath, amsthm, amssymb}
\usepackage{hyperref} 

\newtheorem{thm}{Theorem}
\newtheorem{cor}{Corollary}
\newtheorem{lem}{Lemma}
\newtheorem{prop}{Proposition}

\theoremstyle{remark}

\theoremstyle{definition}
\newtheorem{defi}[thm]{Definition}

\renewcommand{\leq}{\leqslant}
\renewcommand{\geq}{\geqslant}

\newcommand{\IC}{\mathbb{C}}

\newcommand{\IZ}{\mathbb{Z}}
\newcommand{\IN}{\mathbb{N}}

\newcommand{\IFp}{\mathbb{F}_p}
\newcommand{\IF}{\mathbb{F}}
\newcommand{\IFtwo}{{\mathbb{F}_2}}

\newcommand{\VP}{\mathsf{VP}}
\newcommand{\VNP}{\mathsf{VNP}}
\newcommand{\VNPe}{\mathsf{VNP_e}}		
\newcommand{\unifVP}{\mathsf{unif\text{-}VP}}
\newcommand{\unifVNP}{\mathsf{unif\text{-}VNP}}
\newcommand{\asize}{\mathsf{asize}}
\newcommand{\asizez}{\mathsf{asize_0}}

\renewcommand{\P}{\mathsf{P}}
\newcommand{\PH}{\mathsf{PH}}

\newcommand{\poly}{\mathsf{poly}}
\newcommand{\size}{\mathsf{size}}
\newcommand{\sharpP}{\mathsf{\#P}}
\newcommand{\parityP}{\mathsf{\oplus P}}
\newcommand{\ModpP}{\mathsf{Mod}_{p}\mathsf{P}}
\newcommand{\PP}{\mathsf{PP}}

\newcommand{\coRP}{\mathsf{coRP}}

\newcommand{\NP}{\mathsf{NP}}
\newcommand{\MA}{\mathsf{MA}}
\newcommand{\AM}{\mathsf{AM}}
\newcommand{\AMA}{\mathsf{AMA}}
\newcommand{\coNP}{\mathsf{coNP}}
\newcommand{\CH}{\mathsf{CH}}
\newcommand{\CeP}{\mathsf{C_{=}P}}
\newcommand{\GapP}{\mathsf{GapP}}

\newcommand{\HC}{\ensuremath{\text{HC}}}
\newcommand{\HN}{\ensuremath{\text{HN}}}
\newcommand{\per}{\ensuremath{\text{per}}}

\newcommand{\weight}{\omega}

\newcommand{\product}[2]{\underset{#1}{\overset{#2}{\prod}}} 
\newcommand{\summ}[2]{\underset{#1}{\overset{#2}{\sum}}} 
\newcommand{\sumonwords}[2]{\summ{#1 \in \{0,1\}^{#2}}{}} 
\newcommand{\interpol}[3]{\product{i=1}{#3}{#1}_i^{{#2}_i}(1-{#1}_i)^{1-{#2}_i}} 

\title{On fixed-polynomial size circuit lower bounds for uniform polynomials in the sense of Valiant}

\begin{document}

\author{Hervé Fournier
\thanks{Univ Paris Diderot, Sorbonne Paris Cit\'e,
Institut de Math\'ematiques de Jussieu, UMR 7586 CNRS, F-75205 Paris, France.
Email: \texttt{fournier@math.univ-paris-diderot.fr}.}
\and
Sylvain Perifel
\thanks{Univ Paris Diderot, Sorbonne Paris Cit\'{e}, LIAFA, UMR 7089
  CNRS, F-75205 Paris, France.
Email: \texttt{sylvain.perifel@liafa.univ-paris-diderot.fr}.}
\and
Rémi de Verclos
\thanks{ENS Lyon. F-69342 Lyon, France.
Email: \texttt{remi.de\_joannis\_de\_verclos@ens-lyon.fr}.}
}

\maketitle
\begin{abstract}
We consider the problem of fixed-polynomial lower bounds
on the size of arithmetic circuits computing uniform families of polynomials.
Assuming the Generalised Riemann Hypothesis (GRH), we show that for all $k$, there exist polynomials with coefficients in $\MA$ having no arithmetic circuits of size $O(n^k)$ over $\IC$ (allowing any complex constant). We also build a family of polynomials that can be evaluated in $\AM$ having no arithmetic circuits of size $O(n^k)$.
Then we investigate the link between fixed-polynomial
size circuit bounds in the Boolean and arithmetic settings. In characteristic zero, it is proved that $\NP \not\subset \size(n^k)$, or $\MA \subset \size(n^k)$, or $\NP=\MA$ imply lower bounds on the circuit size of uniform polynomials in $n$ variables from the class $\VNP$ over~$\IC$, assuming GRH. In positive characteristic $p$, uniform polynomials in $\VNP$ have circuits of fixed-polynomial size if and only if 
both $\VP = \VNP$ over $\IF_p$ and $\ModpP$ has circuits of fixed-polynomial size.
\end{abstract}


\section{Introduction}

Baur and Strassen~\cite{BaurS83} proved in 1983 that the number of arithmetic
operations needed to compute the polynomials $x_1^n + \ldots + x_n^n$
is $\Omega(n \log n)$.
This is still the best lower bound on uniform
polynomials on $n$ variables and of degree $n^{O(1)}$, if uniformity means
having circuits computed in polynomial time.

If no uniformity condition is required, 
lower bounds for polynomials have been known since Lipton~\cite{Lipton75}.
For example, Schnorr~\cite{Schnorr78}, improving on~\cite{Lipton75} and
Strassen~\cite{Strassen74}, showed for any $k$ a
lower bound $\Omega(n^k)$ on the complexity of a family $(P_n)$ of univariate
polynomials of degree polynomial in $n$ -- even allowing arbitrary complex constants in the circuits.
The starting point of Schnorr's method is to remark that
the coefficients of a polynomial computed by a circuit using
constants $\alpha = (\alpha_1,\ldots,\alpha_p)$ is given by a polynomial
mapping in $\alpha$. Hence, finding hard polynomials reduces to finding
a point outside the image of the mapping associated to some circuit which
is universal for a given size.
This method has been studied and extended by Raz~\cite{Raz10}.

In the Boolean setting, this kind of fixed-polynomial lower bounds has already
drawn a lot of attention, from Kannan's result~\cite{Kannan82} proving that
for all $k$, $\mathsf{\Sigma_2^p}$ does not have circuits of size $n^k$,
to~\cite{FortnowSW09}, delineating the frontier of Boolean classes which are
known to have fixed-polynomial size circuits lower bounds.  It might seem easy
to prove similar lower bounds in the algebraic world, but the fact that
arbitrary constants from the underlying field (e.g. $\IC$) are allowed
prevents from readily adapting Boolean techniques.

Different notions of uniformity can be thought of, either in terms of the circuits computing the
polynomials, or in terms of the complexity of computing the coefficients. For
instance, an inspection of the proof of Schnorr's result mentioned above
shows that the coefficients of the polynomials can be computed
in exponential time. But this complexity is generally considered
too high to qualify these polynomials as uniform.

The first problem we tackle is the existence of hard polynomials (i.e. without small circuits over $\IC$) but with coefficients that are
``easy to compute''.
The search for a uniform family of polynomials with no circuits of size $n^k$
was pursued recently by Jansen and Santhanam~\cite{JansenS12}. They show in
particular that there exist polynomials with coefficients in $\MA$ (thus,
uniform in some sense) but not computable by arithmetic circuits of size $n^k$
over $\IZ$.\footnote{Even though this result is not stated explicitly in their
  paper, it is immediate to adapt their proof to our context.} Assuming the
Generalised Riemann Hypothesis (GRH), we extend their result to the case of
circuits over the complex field.
GRH is used to eliminate the complex constants in the circuits, by considering
solutions over $\IFp$ of systems of polynomial equations, for a small prime
$p$, instead of solutions over $\IC$.
In fact, the family of polynomials built by Jansen and Santhanam is also uniform in the following way: it can be evaluated at integer points in $\MA$. Along this line, we obtain families of polynomials without arithmetic circuits of size $n^k$ over $\IC$ and that can be evaluated in $\AM$.
The arbitrary complex constants prevents us to readily adapt Jansen and
Santhanam's method and we need to use in addition the $\AM$ protocol
of Koiran~\cite{Koiran96} in order to decide whether a system of polynomial equations
has a solution over~$\IC$.

Another interesting and robust notion of uniformity is provided by Valiant's
algebraic class $\VNP$, capturing the complexity of the permanent. The usual
definition is non-uniform, but a natural uniformity condition can be required
and gives two equivalent characterisations: in terms of the uniformity of
circuits and in terms of the complexity of the coefficients. This is one of the
notions we shall study in this paper and which is also used by Raz~\cite{Raz10} (where the term {\it explicit} is used to denote uniform families of $\VNP$ polynomials).
The second problem we study is therefore to give an $\Omega(n^k)$ lower bound on the complexity of
an $n$-variate polynomial in the uniform version of the class $\VNP$.
Note that from Valiant's criterion, it corresponds to the coefficients being in $\GapP$, so it is a special case of coefficients that are easy to compute.
 Even though $\MA$ may seem a small class in comparison with
$\GapP$ (in particular due to Toda's theorem $\PH\subseteq\P^{\sharpP}$),
the result obtained above does not yield lower bounds for the uniform version of $\VNP$.

We show how fixed-polynomial circuit size lower bound on uniform $\VNP$ is
connected to various questions in Boolean complexity. For instance, the
hypothesis that $\NP$ does not have circuits of size $n^k$ for all $k$,
or the hypothesis that $\MA$ has circuits of size $n^k$ for some
$k$, both imply the lower bound on the uniform version of $\VNP$
assuming GRH. Concerning
the question on finite fields, we show
an equivalence between lower bounds on uniform $\VNP$ and
standard problems in Boolean and algebraic complexity.

The paper is organised as follows. Definitions, in particular of the uniform versions of Valiant's classes, are given in Section~\ref{sec:prelim}.
Hard families of polynomials with easy to compute coefficients, or that are easy to evaluate, are built in Section~\ref{sec:schnorr}.
Finally, conditional lower bounds on uniform $\VNP$
are presented in the last section.

\section{Preliminaries}\label{sec:prelim}

\subsubsection*{Arithmetic circuits}

An arithmetic circuit over a field $K$ is a directed acyclic graph
whose vertices have indegree 0 or 2 and where a single vertex (called the
output) has outdegree 0. Vertices of indegree 0 are called inputs and are
labelled either by a variable $x_i$ or by a constant $\alpha\in K$. Vertices
of indegree 2 are called gates and are labelled by $+$ or $\times$.

The polynomial computed by a vertex is defined recursively as follows: the
polynomial computed by an input is its label; a $+$ gate (resp. $\times$
gate), having incoming edges from vertices computing the polynomials $f$ and
$g$, computes the polynomial $f+g$ (resp. $fg$). The polynomial computed by a
circuit is the polynomial computed by its output gate.

A circuit is called \emph{constant-free} if the only constant
appearing at the inputs is $-1$. The \emph{formal degree} of a circuit is
defined by induction in the following way: the formal degree of a leaf is $1$,
and the formal degree of a sum (resp. product) is the maximum (resp. sum) of the formal degree of the incoming subtrees (thus constants ``count as variables'' and there is no possibility of cancellation).

We are interested in sequences of arithmetic circuits $(C_n)_{n\in\IN}$,
computing sequences of polynomials $(P_n)_{n\in\IN}$ (we shall usually drop
the subscript ``$n\in\IN$'').

\begin{defi}
Let $K$ be a field. If $s:\IN\to\IN$ is a function, a family $(P_n)$ of
polynomials over $K$ is in $\asize_K(s(n))$ if it is computed by a family of
arithmetic circuits of size $O(s(n))$ over $K$.

Similarly, $\size(s(n))$ denotes the set of (Boolean) languages
decided by Boolean circuits of size $O(s(n))$.
\end{defi}

\subsubsection*{Counting classes}

A function $f:\{0,1\}^\star\to\IN$ is in $\sharpP$ if there exists a
    polynomial $p(n)$ and a language $A\in\P$ such that for all
    $x\in\{0,1\}^\star$
    $$f(x)=|\{y\in\{0,1\}^{p(|x|)},\ (x,y)\in A\}|.$$
A function $g:\{0,1\}^\star\to\IZ$ is in $\GapP$ if there exist two
    functions $f,f'\in\sharpP$ such that $g=f-f'$.
The class $\CeP$ is the set of languages $A=\{x,\ g(x)=0\}$ for some
    function $g\in\GapP$.
The class $\parityP$ is the set of languages $A=\{x,\ f(x)\mbox{ is odd}\}$
for some
    function $f\in\sharpP$.
We refer the reader to~\cite{HemaspaandraO02} for more details on counting classes.

\subsubsection*{Valiant's classes and their uniform counterpart}

Let us first recall the usual definition of Valiant's classes.
\begin{defi}[Valiant's classes] Let $K$ be a field.
A family $(P_n)$ of polynomials over $K$ is in the class $\VP_K$ if
    the degree of $P_n$ is polynomial in $n$ and $(P_n)$ is computed by a
    family $(C_n)$ of \emph{polynomial-size} arithmetic circuits over $K$.

A family $(Q_n(x))$ of polynomials over $K$ is in the class
    $\VNP_K$ if there exists a family $(P_n(x,y))\in \VP_K$ such that
    $$Q_n(x)=\sum_{y\in\{0,1\}^{|y|}} P_n(x,y).$$
\end{defi}

The size of $x$ and $y$ is limited by the circuits for $P_n$ and
is therefore polynomial.
Note that the only difference between $\VP_K$ and $\asize_K(\poly)$ is the constraint on
the degree of $P_n$.
If the underlying field $K$ is clear, we shall drop the subscript ``$K$''
  and speak only of $\VP$ and $\VNP$.
 Based on these usual definitions, we now define uniform versions of Valiant's classes.

\begin{defi}[Uniform Valiant's classes]
  Let $K$ be a field. A family of circuits $(C_n)$ is called uniform if the
  (usual, Boolean) encoding of $C_n$ can be computed in time $n^{O(1)}$. A
  family of polynomials $(P_n)$ over $K$ is in the class $\unifVP_K$ if it is
  computed by a uniform family of constant-free arithmetic circuits of
  polynomial formal degree.

  A family of polynomials $(Q_n(x))$ over $K$ is in the class
  $\unifVNP_K$ if $Q_n$ has $n$ variables $x=x_1,\dots,x_n$ and there exists a
  family $(P_n(x,y)) \in \unifVP_K$ such that
    $$Q_n(x)=\sum_{y \in\{0,1\}^{|y|}} P_n(x,y).$$
\end{defi}

The uniformity condition implies that the size of the circuit
$C_n$ in the definition of $\unifVP$ is polynomial in $n$.
Note that $\unifVP_K$ and $\unifVNP_K$ only depend on the characteristic of the field $K$ (indeed, since no constant from $K$ is allowed in the circuits, these  classes are equal to the ones defined over the prime subfield of $K$).

  In the definition of $\unifVNP$, we have chosen to impose that $Q_n$ has $n$
  variables because this enables us to give a very succinct and clear
  statement of our questions. This is \emph{not} what is done in the usual non-uniform
  definition where the number of variables is only limited by the (polynomial) size of the
  circuit.

The well-known ``Valiant's criterion'' is easily adapted to the uniform case
in order to obtain the following alternative characterisation of
$\unifVNP$.

\begin{prop}[Valiant's criterion]\label{prop:valiant-crit}
In characteristic zero,
a family $(P_n)$ is in $\unifVNP$ iff $P_n$ has $n$ variables, a polynomial
degree and its coefficients are computable in $\GapP$; that is, the function
mapping $(c_1,\ldots,c_n)$ to the coefficient of $X_1^{c_1}\cdots X_n^{c_n}$
in $P_n$ is in $\GapP$.

The same holds in characteristic $p>0$ with coefficients in ``$\GapP \mod p$''\footnote{This is equivalent to the fact that for all $v \in \IFp$, the set of monomials having coefficient $v$ is in $\ModpP$.}.
\end{prop}

Over a field $K$, a polynomial $P(x_1,\dots,x_n)$ is said to be a projection
 of a polynomial $Q(y_1,\dots,y_m)$ if $P(x_1,\dots,x_n)=Q(a_1,\dots,a_m)$
 for some choice of $a_1,\dots,a_m\in\{x_1,\dots,x_n\}\cup K$. A family $(P_n)$ reduces to $(Q_n)$ (via projections) if $P_n$ is a projection of $Q_{q(n)}$ for some polynomially bounded function $q$.

The Hamiltonian Circuit polynomials are defined by
$$\HC_n(x_{1,1},\dots,x_{n,n})=\sum_\sigma\prod_{i=1}^n x_{i,\sigma(i)},$$
where the sum is on all cycles $\sigma\in S_n$ (i.e. on all
the Hamiltonian cycles of the complete graph over $\{1,\ldots,n\}$).
The family $(\HC_n)$
is known to be $\VNP$-complete over any field~\cite{Valiant79} (for
projections). 

\subsubsection*{Elimination of complex constants in circuits}

The weight of a polynomial $P \in \IC[X_1,\ldots,X_n]$ is the sum of the absolute values of its coefficients.
We denote it by $\weight (P)$. It is well known that $\weight$ is a norm of algebra, that is: for $P, Q \in \IC[X_1,\ldots,X_n]$ and $\alpha \in \IC$,
it holds that $\weight (PQ) \leqslant \weight (P)\weight (Q)$,
$\weight (P+Q) \leqslant \weight (P)+\weight (Q)$ and $\weight (\alpha P) = |\alpha| \weight (P)$.

The following result gives a bound on the weight of a polynomial computed by a circuit.

\begin{lem}\label{lem:weight}
Let $P$ be a polynomial computed by an arithmetic circuit of size $s$ and formal degree $d$ with constants of absolute value bounded by $M \geqslant 2$, then $\weight (P) \leqslant M^{s \cdot d}$.
\end{lem}

\begin{proof}
  We prove it by induction on the structure of the circuit $C$ which computes
  $P$. The inequality is clear if the output of $C$ is a constant or a
  variable since $\weight (P) \leqslant M$, $s \geqslant 1$ and $d \geqslant
  1$ in this case. If the output of $P$ is a $+$ gate then $P$ is the sum of
  the value of two polynomials $P_1$ and $P_2$ calculated by subcircuits of
  $C$ of formal degree at most $d$ and size at most $s-1$.
  By induction hypothesis, we have $\weight (P_1) \leqslant
  M^{d(s-1)}$ and $\weight (P_1) \leqslant M^{d(s-1)}$. We have $\weight (P)
  \leqslant \weight (P_1) + \weight (P_2)$ so $\weight (P) \leqslant 2 \cdot
  M^{d(s-1)} \leqslant M^{d(s-1)+1} \leqslant M^{ds}$. If the output of $C$ in
  a $\times$ gate, $P$ is the product some polynomials $P_1$ and $P_2$ each
  calculated by circuits of size at most $s-1$ and degrees $d_1$ and $d_2$
  respectively such that $d_1 + d_2 = d$. Then
  $\weight (P) \leqslant \weight (P_1)\weight (P_2) \leqslant
  M^{(s-1)d_1}M^{(s-1)d_2} = M^{(s-1)d} \leqslant M^{sd}$.
\end{proof}

For $a \in \IN$, we denote by $\pi (a)$ the number of prime numbers smaller than or equal to $a$.
For a system $S$ of polynomial equations with integer coefficients, we denote by $\pi_S(a)$
the number of prime numbers $p \leqslant a$ such that $S$ has a solution over
$\mathbb{F}_p$. The following lemma will be useful for eliminating constants
from $\IC$. (Note that the similar but weaker statement first shown by Koiran~\cite{Koiran96} as a step in his proof of Theorem~\ref{thm:koiran} would be enough for our purpose.)

\begin{lem}[B\"urgisser~{\cite[p.~64]{Burgisser00}}]\label{lem:burgisser}
  Let $S$ be a system of polynomial equations $$P_1(x)=0, \dots, P_m(x)=0$$ with coefficients in $\IZ$ and with the following parameters : $n$
  unknowns, and for all $i$, degree of $P_i$ at most $d$ and $\weight(P_i)\leqslant w$.  If the
  system $S$ has a solution over $\IC$ then under GRH,
  $$\pi_S(a) \geqslant \frac{\pi(a)}{d^{O(n)}} - \sqrt{a} \log(w a).$$
\end{lem}

At last, we need a consequence of $\VNP$ having small arithmetic circuits over the complex field.
\begin{lem}\label{lem:collapse}
Assume GRH. If $\VP = \VNP$ over $\IC$, then $\CH=\MA$.
\end{lem}
\begin{proof}
Assume $\VP = \VNP$ over $\IC$. From the work on Boolean parts of Valiant's classes~\cite[Chapter 4]{Burgisser00}, this implies
$\P / \poly = \PP / \poly = \CH / \poly$, therefore $\MA = \CH$~\cite{LundFKN90}.
\end{proof}

\section{Hard polynomials with coefficients in $\MA$}\label{sec:schnorr}

We begin with lower bounds on polynomials with coefficients in $\PH$ before
bringing them down to $\MA$.

\subsubsection*{Hard polynomials with coefficients in $\PH$}

We first need to recall a couple of results. The first one 
is an upper bound on the complexity of the following problem called $\HN$:
\begin{description}
\item[\it Input] A system $S=\{P_1=0,\dots,P_m=0\}$ of $n$-variate
    polynomial equations with integer coefficients, each polynomial
    $P_i\in\IZ[x_1,\dots,x_n]$ being given as a constant-free arithmetic
    circuit.
\item[\it Question] Does the system $S$ have a solution over $\IC^n$?
\end{description}

\begin{thm}[Koiran~\cite{Koiran96}]\label{thm:koiran}
  Assuming GRH is true, $\HN\in\PH$.
\end{thm}

Koiran's result is stated here for polynomials given by arithmetic circuits,
instead of the list of their coefficients. Adapting the result of the original
paper in terms of arithmetic circuits is not difficult: it is enough to add
one equation per gate expressing the operation made by the gate, thus
simulating the whole circuit.

The second result is used in the proof of Schnorr's result mentioned in the
introduction.

\begin{lem}[Schnorr~\cite{Schnorr78}]\label{lem:schnorr}
  Let $(U_n)$ be the family of polynomials defined inductively as follows:
  $$\begin{cases}
    U_1=a^{(1)}_0+b^{(1)}_0x \quad \text{where $a^{(1)}_0,b^{(1)}_0$ and $x$
      are new variables}\\
    U_n=\left(\sum_{i=1}^{n-1} a^{(n)}_i
      U_i\right)\left(\sum_{i=1}^{n-1} b^{(n)}_i U_i\right)\quad \text{where
    }a^{(n)}_i,b^{(n)}_i\text{ are new variables.}
  \end{cases}$$
  Thus $U_n$ has variables $x$, $a^{(j)}_i$ and $b^{(j)}_i$ (for $1\leqslant
  j\leq n$ and $0\leqslant i<j$). For simplicity, we will write $U_n(a,b,x)$,
  where the total number of variables in the tuples $a,b$ is $n(n+1)$.

    For every univariate polynomial $P(x)$ over $\IC$ computed by an
    arithmetic circuit of size $s$, there are constants $a,b\in\IC^{s(s+1)}$
    such that $P(x)=U_s(a,b,x)$.
\end{lem}

The polynomials $U_s$ in this lemma are universal in the sense that they
can simulate any circuit of size $s$; the definition of such a polynomial
indeed reproduces the structure of an arbitrary circuit by letting at each
gate the choice of the inputs and of the operation, thanks to new variables.

The third result we'll need is due to Hrube\v{s} and Yehudayoff~\cite{HrubesY11}
and relies on Bézout's Theorem. Showing Theorem~\ref{thm:coef-ph} could also
be done without using algebraic geometry, but this would complicate the
overall proof.

\begin{lem}[Hrube\v{s} and Yehudayoff~\cite{HrubesY11}]\label{lem:hy}
  Let $F:\IC^n\to\IC^m$ be a polynomial map of degree $d>0$, that is,
  $F=(F_1,\dots,F_m)$ where each $F_i$ is a polynomial of degree at most $d$. Then $|F(\IC^n)\cap\{0,1\}^m| \leqslant (2d)^n$.  
\end{lem}

We are now ready to give our theorem.

\begin{thm}\label{thm:coef-ph}
  Assume GRH is true. For any constant $k$, there is a family $(P_n)$ of
  univariate polynomials with coefficients in $\{0,1\}$ satisfying:
  \begin{itemize}
  \item $\deg(P_n)=n^{O(1)}$ (polynomial degree);
  \item the coefficients of $P_n$ are computable in $\PH$, that is, on input
    $(1^n,i)$ we can decide in $\PH$ if the coefficient of $x^i$ is $1$;
  \item $(P_n)$ is not computed by arithmetic circuits over $\IC$ of size
    $n^k$.
  \end{itemize}
\end{thm}

\begin{proof}
  Fix $s=n^k$. Consider the universal polynomial $U_s(a,b,x)$ of
  Lemma~\ref{lem:schnorr} simulating circuits of size $s$. If $\alpha_i^{(s)}$
  denotes the coefficient of $x^i$ in $U_s$, then we have the relation
  $$\alpha_i^{(s)}=\sum_{i_1+i_2=i \atop s_1,s_2<s}
  a_{s_1}^{(s)}b_{s_2}^{(s)}\alpha_{i_1}^{(s_1)}\alpha_{i_2}^{(s_2)}.$$
  By induction, the coefficient $\alpha_i^{(s)}$ is therefore a polynomial in
  $a,b$ of degree $\leqslant (i+1)2^{2s}$.

  Now, we would like to find a polynomial whose coefficients are different
  from the $\alpha_i^{(s)}$ for any value of $a,b$. This will be done thanks
  to Lemma~\ref{lem:hy}, but we have to use it in a clever way because our
  method requires to use interpolation on $d+1$ points to identify two
  polynomials of degree $d$: hence we need to ``truncate'' the polynomial
  $U_s$ to degree $d$.

  Fix $d=s^4$. It follows from the beginning of the proof that the map
  computing the first $(d+1)$ coefficients of $U_s$
  $$\begin{array}{llcl}
    F: & \IC^{s(s+1)} & \to & \IC^{d+1}\\
    & (a,b) & \mapsto & (\alpha_0^{(s)},\dots,\alpha_d^{(s)})
  \end{array}$$
  is a polynomial map of degree at most $(d+1)2^{2s}$. Since
  $((d+1)2^{2s})^{s(s+1)}<2^{d+1}$, by Lemma~\ref{lem:hy} there exist
  coefficients $(\beta_0,\dots,\beta_d)\in\{0,1\}^{d+1}$ not in
  $F(\IC^{s(s+1)})$. In other words, for any values of $a,b$ in
  $\IC$, the first $(d+1)$ coefficients of $U_s$ differ from
  $(\beta_0,\dots,\beta_d)$.

  Let $P_{\beta}(x)$ be the polynomial $\sum_{i=0}^d\beta_ix^i$ and let us
  call $U_{s|_d}$ the truncation of $U_s$ up to degree $d$, that is, the sum
  of all the monomials of degree $\leqslant d$ in $x$. For any instantiation
  of $a,b$ in $\IC$, we have $U_{s|_d}(a,b,x)\neq
  P_{\beta}(x)$. Since both polynomials are of degree smaller than or equal to
  $d$, this means that there exists an integer $m\in\{0,\dots,d\}$ such that
  $U_{s|_d}( a, b,m)\neq P_{\beta}(m)$. Therefore the following
  system of polynomial equations with unknowns $ a, b$:
  $$S_{\beta}=\{U_{s|_d}( a, b,m)=P_{\beta}(m)\ :\
  m\in\{0,\dots,d\}\}$$
  has no solution over $\IC$.

  Conversely, consider now this system for other coefficients than
  $\beta$, that is, $S_{\gamma}$ for
  $\gamma_0,\dots,\gamma_d\in\{0,1\}$. If $S_{\gamma}$ does not have a
  solution over $\IC$, this means that for any instantiation of $ a,
  b\in\IC$ we have $U_{s|_d}( a, b,x)\neq P_{\gamma}(x)$, hence
  $P_{\gamma}$ is not computable by a circuit of size $s$ by
  Lemma~\ref{lem:schnorr}.

  The goal now is then to find values of $\gamma\in\{0,1\}^{d+1}$ such that
  $S_{\gamma}$ does not have a solution over $\IC$.

  Remark first that on input $\gamma_0,\dots,\gamma_d\in\{0,1\}$ and
  $m\in\{0,\dots,d\}$, we can describe in polynomial time a circuit
  $C_{\gamma,m}( a, b)$ computing the polynomial $U_{s|_d}(
  a, b,m)-P_{\gamma}(m)$. Indeed, $U_s$ is computable by an easily
  described circuit following its definition, hence its truncation to degree
  $d$ also is, and a circuit for $P_{\gamma}$ is also immediate if we are
  given $\gamma$. Therefore, we can describe in polynomial time the system
  $S_{\gamma}$ to be used in Theorem~\ref{thm:koiran}.

  The algorithm in $\PH$ to compute the coefficients of a polynomial
  $P_{\beta}$ without circuits of size $s$ is then the following on input
  $(1^n,i)$:
  \begin{itemize}
  \item Find the lexicographically first $\gamma_0,\dots,\gamma_d\in\{0,1\}$
    such that $S_{\gamma}\not\in\HN$;
  \item accept iff $\gamma_i=1$.
  \end{itemize}
  This algorithm is in $\PH^\HN$. By Theorem~\ref{thm:koiran}, if we assume
  GRH then the problem $\HN$ is in $\PH$. We deduce that computing the
  coefficients of $P_{\gamma}$ can be done in $\PH$.
\end{proof}

\subsubsection*{Hard polynomials with coefficients in $\MA$}

Allowing $n$ variables instead of only one, we can even obtain lower
bounds for polynomials with coefficients in $\MA$.

\begin{cor}\label{cor:coef-ma}
  Assume GRH is true. For any constant $k$, there is a family $(P_n)$ of
  polynomials on $n$ variables, with coefficients in $\{0,1\}$, of degree
  $n^{O(1)}$, with coefficients computable in $\MA$, and such that $(P_n) \not\in \asize_\IC (n^k)$.
\end{cor}
\begin{proof}
  If the Hamiltonian family $(\HC_n)$ does not have circuits of polynomial
  size over $\IC$, consider the following variant of a family with $n$
  variables: $\HC'_n(x_1,\dots,x_n)=
  \HC_{\lfloor\sqrt{n}\rfloor}(x_1,\dots,\allowbreak x_{\lfloor\sqrt{n}\rfloor^2})$. This
  is a family whose coefficients are in $\P$ (hence in $\MA$) and without
  circuits of size $n^k$.

  On the other hand, if the Hamiltonian family $(\HC_n)$ has circuits of
  polynomial size over $\IC$, then $\PH=\MA$ by Lemma~\ref{lem:collapse}. Therefore the family
  of polynomials of Theorem~\ref{thm:coef-ph} has its coefficients in $\MA$.
\end{proof}

\subsubsection*{Hard polynomials that can be evaluated in $\AM$}

A family of polynomials $(P_n(x_1,\dots,x_n))$ is
said to be evaluable in $\AM$ if the language
$$\{(x_1,\dots,x_n,i,b)\ |\ \text{the $i$-th bit of }P_n(x_1,\dots,x_n)\text{
  is }b\}$$
is in $\AM$, where $x_1, \ldots, x_n, i$ are integers given in binary and $b \in \{0,1\}$.
In the next proposition, we show how to obtain polynomials which can
be evaluated in $\AM$.
The method is based on Santhanam~\cite{Santhanam09} and Koiran~\cite{Koiran96rr}.

\begin{prop}\label{prop:eval-am}
  Assume GRH is true. For any constant $k$, there is a family $(P_n)$ of
  polynomials on $n$ variables, with coefficients in $\{0,1\}$, of degree
  $n^{O(1)}$, evaluable in $\AM$ and such that $(P_n) \not\in \asize_\IC (n^k)$.
\end{prop}

\begin{proof}
  We adapt the method of Santhanam~\cite{Santhanam09} to the case of circuits with complex constants.

  If the permanent has polynomial-size circuits over $\IC$, then $\PH=\MA$
by Lemma~\ref{lem:collapse} and
  hence the family of polynomials of Theorem~\ref{thm:coef-ph} is evaluable in $\MA \subseteq \AM$.

  Otherwise, call $s(n)$ the minimal size of a circuit over $\IC$ for
  $\per_n$. The $n$-tuple of variables $(x_1,\dots,x_n)$ is split in two parts
  $(y,z)$ in the unique way satisfying $0<|y|\leq |z|$ and $|z|$ a power of
  two. Remark therefore that $|y|$ can take all the values from 1 to $|z|$
  depending on $n$. We now define the polynomial $P_n(y,z)$:
  $$\begin{cases}
     P_n(y,z)=\per(y) & \text{if }|y|\text{ is a
      square and }s(\sqrt{|y|}) \leq n^{2k}\\
    P_n(y,z)=0 & \text{otherwise.}
  \end{cases}$$

  Let us first show that $(P_n)$ does not have circuits of size $n^k$. By
  hypothesis there exist infinitely many $n$ such that $s(n)>(3n^2)^{2k}$: let
  $n_0$ be one of them and take $m$ the least power of two such that
  $s(n_0)\leq (m+n_0^2)^{2k}$, which implies $m\geq 2n_0^2$. Let
  $n_1=m+n_0^2$: by definition of $(P_n)$, we have
  $P_{n_1}(y,z)=\per_{n_0}(y)$. By definition of $m$, $s(n_0)>
  (m/2+n_0^2)^{2k}>(n_1/2)^{2k}>n_1^k$. This means that $\per_{n_0}$, and
  hence $P_{n_1}$, does not have circuits of size $n_1^k$.

  We now show that $(P_n)$ can be evaluated in $\AM$. We give an $\AMA$ protocol which is enough since $\AMA=\AM$ (see~\cite{BalcazarDG90}).
  
The protocol described below heavily relies on the technique used in~\cite[Theorem 2]{Koiran96rr} to prove that $\HN \in \AM$.

In the following, we need to test if $\per_t$ (for some $t$) has an arithmetic circuit of size $s$ over the complex field. If this is true, Merlin can give the skeleton of the circuit but he cannot give the complex constants. Hence, he gives
a circuit $C(y,u)$ where $y$ is the input (of size $t \times t$) and $u$ a tuple of formal variables. Consider the following system $S$ : for all $\varepsilon \in \{0,\ldots,2^s\}^{|y|}$, take the equation $C(\varepsilon,u)=\per_t(\varepsilon)$.
For some values $\alpha \in \IC^{|u|}$, the degree of the polynomial computed by the circuit $C(y,\alpha)$ is at most $2^s$; hence,
the system $S$ is satisfiable over $\IC$ iff the variables $u$ can be replaced by complex numbers $\alpha$ such that $C(y,\alpha)$ computes the permanent over the complex field.

The system $S$ has the following parameters: the number of variables is $|u|$ which is at most $s$, the degree of each equation is bounded by $2^s$, the number of equations is $2^{O(s^2)}$ and the bitsize of each coefficient is $2^{s^{O(1)}}$. Hence, by~\cite[Theorem 1]{Koiran96rr}, there is an integer $m=s^{O(1)}$ and $x_0 =2^{s^{O(1)}}$ such that the following holds.
Let $E$ be the set of primes $p$ smaller than $x_0$ such that $S$ has a solution modulo $p$.
\begin{itemize}
\item If $S$ is not satisfiable over $\IC$, then $|E| \leq 2^{m-2}$;
\item If $S$ is satisfiable over $\IC$, then $|E| \geq m2^m$.
\end{itemize}
Testing if $|E|$ is large or small is done via the following probabilistic argument. For some matrices $A_j$
over $\IF_2$, the predicate
$\phi(A_1,\ldots,A_m)$ is defined  as
$$\exists p_0,p_1,\ldots,p_m \in E\ :\ \psi(A_1,\ldots,A_m,p_0,\ldots,p_m)$$
where
$$ \psi(A_1,\ldots,A_m,p_0,\ldots,p_m)  \equiv \bigwedge_{j=1}^m
\left( A_j p_0 = A_j p_j \land p_0 \neq p_j \right).$$
If $A_j$ are seen as hashing functions, the predicate $\phi$ above expresses that there are enough collisions between elements of $E$.
Based on~\cite{Stockmeyer85}, it is proved in~\cite{Koiran96rr} that if $|E| \leq 2^{m-2}$, the probability that $\phi(A_1,\ldots,A_m)$ holds is at most $1/2$ when the matrices $A_j$ are chosen uniformly at random, whereas it is $1$ when $|E| \geq m2^m$.

We are now ready to explain the $\AMA$ protocol to evaluate the family $(P_n)$. On input $(x_1,\dots,x_n,i,b)$, the $\AMA$ protocol is the following:
  \begin{itemize}
  \item Arthur splits $(x_1,\dots,x_n)$ in $(y,z)$ in the unique way.
  If $|y|$ is not a square, he accepts if $b=0$ and rejects if $b \neq 0$.
Otherwise, call $t=\sqrt{|y|}$; Arthur sends to Merlin random matrices $A_1, \ldots, A_m$ over $\IFtwo$.
  \item  Merlin sends to Arthur the skeleton $C(y,u)$ of a circuit of size $\leq n^{2k}$ supposedly computing $\per_t$ over $\IC$ (that is, the circuit with complex constants replaced with formal variables $u$). He also sends 
  prime integers $p_0, \ldots,p_m$
  together with constants $\alpha_{p_j} \in \IF_{p_j}^{|u|}$ for $C$, for all $0 \leq j \leq m$. He also sends a prime number $p \geq n! M^n$ 
(where $M$ is the largest value in $(x_1,\dots,x_n)$) and constants of $\alpha_p$ over $\IF_p$ for $C$.
  \item Arthur checks that $p_0, \ldots, p_m$ produce a collision (that is,
    that $\psi(A_1,\dots,A_m,p_0,\dots,p_m)$ is true). Then he checks that all $p_j$ and $p$ are primes and that the circuits
  $C(y,\alpha_{p_j})$ and $C(y,\alpha_p)$
   compute the permanent modulo $p_0, \ldots, p_m, p$ (using the $\coRP$ algorithm of~\cite{KabanetsI04}).
If any of these tests fails, Arthur accepts iff $b=0$. Otherwise, he computes $C(y,\alpha_p)$ and accepts iff its $i$-th  bit is equal to $b$.
  \end{itemize}

  If $(y,z)$ is such that $|y|$ is a square and $s(|y|)\leq n^{2k}$, then
  $P_n(y,z)=\per(y)$. We show that Merlin can convince Arthur with probability 1. Merlin sends a correct skeleton $C$: since $|E|\geq m2^m$, there are prime integers $p_0,\dots,p_m\in E$ such that $\psi(A_1,\dots,A_m,p_0,\dots,p_m)$ holds. Merlin sends such numbers $p_j$ and $p$ together with the correct constants for the circuit $C$ to compute the permanent modulo $p_j$ and $p$. In the third round, all the verifications are satisfied with probability 1 and Arthur gives the right answer.

  On the other hand, if $|y|$ is not a square then whatever Merlin sends,
  Arthur accepts only if $b=0$, which is the right answer. Assume now that
  $s(|y|)> n^{2k}$; then $|E|\leq 2^{m-2}$. Whatever Merlin sends as prime
  numbers $p_j$, the probability (over the matrices $A$) that all $p_j$ belong
  to $E$ and produce a collision is at most $1/2$. Since the error when
  testing if $p_j \in E$ can be made as small as we wish (testing if
  $C(y,\alpha_{p_j})$ computes $\per(y)\mod p_j$ is done in $\coRP$), the probability that the whole protocol gives the wrong answer in this case is bounded by $2/3$.
\end{proof}

\section{Conditional lower bounds for uniform $\VNP$}

\subsubsection*{In characteristic zero}

In this whole section we assume GRH is true. Our main result in this section
is that if for all $k$, $\CeP$ has no circuits of size $n^k$,
then the same holds for $\unifVNP$ (in characteristic zero). For the clarity of exposition, we first
prove the weaker result where the assumption is on the class $\NP$ instead.

\begin{lem}\label{lem:NP}
If there exists $k$ such that $\unifVNP \subset \asize_\IC(n^k)$,
then there exists $\ell$ such that $\NP \subset \size(n^{\ell})$.
\end{lem}

\begin{proof}
Let us assume that $\unifVNP \subseteq \asize_\IC(n^k)$.
Let $L \in \NP$. There is a polynomial $q$ and a polynomial time
computable relation $\phi : \{0,1\}^* \times \{0,1\}^* \rightarrow \{0,1\}$ such that for all $x \in \{0,1\}^n$,
$x \in L$ if and only if $\exists y \in \{0,1\}^{q(n)}\ \phi(x,y)=1$.

We define the polynomial $P_n$ by
 $$P_n(X_1,\ldots,X_n) = \sumonwords{x}{n} \left(
 \sumonwords{y}{q(n)} \phi(x,y) \right) \interpol{X}{x}{n}.$$
Note that for $x \in \{0,1\}^n$, $P_n(x)$ is the number of elements $y$ in relation with $x$ via $\phi$.
By Valiant's criterion (Proposition~\ref{prop:valiant-crit}),
the family $(P_n)$ belongs to
$\unifVNP$ in characteristic $0$. By hypothesis, there exists
a family of arithmetic circuits $(C_n)$ over $\IC$ computing $(P_n)$, with
$C_n$ of size $t=O(n^k)$.

Let $\alpha=(\alpha_1, \ldots, \alpha_t)$ be the complex constants used by the
circuit. We have
$P_n(X_1,\ldots,X_n) = C_n (X_1,\ldots,X_n,\alpha)$.
Take one unknown $Y_i$ for each $\alpha_i$ and one additional unknown $Z$,
and consider the following system $S$:

$$\left\{
 \begin{array}{l}
 	\left( \prod_{x \in L \cap \{0,1\}^n} C_n(x, Y) \right) \cdot Z = 1\\
  	C_n(x, Y) = 0 \text{ for all } x \in \{0,1\}^n \setminus L.
 \end{array}
\right.$$

Note that introducing one equation for each $x \in L \cap \{0,1\}^n$
(as we did for each $x \in \{0,1\}^n \setminus L$) would not work
since it would require to introduce an exponential number of new variables.
 
Let $\beta = \left( \prod_{x \in L \cap \{0,1\}^n} C_n(x, \alpha)
\right)^{-1}$. Then $(\alpha,\beta)$ is a solution of $S$ over $\IC$.

The system $S$ has $t+1=O(n^k)$ unknowns. The degree of $C_n(x,Y)$ is bounded
by $2^t$; hence the degree of $S$ is at most $2^{O(n^k)}$.
Moreover, the weight of the polynomials in $S$ is bounded by
$2^{2^{O(n^k)}}$ using Lemma~\ref{lem:weight}.

Since the system $S$ has the solution $(\alpha,\beta)$
over $\IC$, by Lemma~\ref{lem:burgisser}
it has a solution over $\IFp$ for some
$p$ small enough. We recall that $\pi(p) \sim p / \log p$;
hence the system $S$ has a solution
over $\IFp$ for $p = 2^{O(n^{2k})}$. 

Consider $p$ as above and $(\alpha',\beta')$ a solution
of the system $S$ over $\IFp$.
By definition of $S$, when the circuit $C_n$ is evaluated over $\IFp$, the following is satisfied:
$$\begin{cases}
  \forall x\in L \cap \{0,1\}^n,& C_n(x,\alpha') \neq 0,\\
  \forall x\in \{0,1\}^n\setminus L,& C_n(x,\alpha')= 0.\\
\end{cases}$$
Computations over $\IFp$ can
be simulated by Boolean circuits, using $\log_2 p$ bits to represent
an element of $\IFp$, and $O(\log^2 p)$ gates
to simulate an arithmetic operation.
This yields Boolean circuits of size $n^\ell$ for $\ell = O(k)$
to decide the language $L$.
\end{proof}

\begin{thm}\label{thm:zero-char}
Assume GRH is true. Suppose one of the following conditions holds:
\begin{enumerate}
\item $\NP \not\subset \size(n^k)$ for all $k$;
\item $\CeP \not\subset \size(n^k)$ for all $k$;
\item $\MA \subset \size(n^k)$ for some $k$;
\item $\NP=\MA$.
\end{enumerate}
Then $\unifVNP \not\subset \asize_\IC(n^k)$ for all $k$.
\end{thm}

\begin{proof}
The first point is proved in Lemma~\ref{lem:NP}.

The second point subsumes the first since $\coNP \subseteq \CeP$. It can be
proved in a very similar way. Indeed consider $L \in \CeP$ and $f\in\GapP$
such that $x\in L\iff f(x)=0$, and its associated family of polynomials
 $$P_n(X_1,\ldots,X_n) = \sumonwords{x}{n} f(x) \interpol{X}{x}{n}$$
as in the proof of Lemma~\ref{lem:NP}. Then for all $x\in\{0,1\}^n$,
$P_n(x)=0$ iff $x\in L$.
The family $(P_n)$ belongs to $\unifVNP$ and thus,
assuming $\unifVNP \subset \asize_\IC(n^k)$, has arithmetic circuits $(C_n)$
over $\IC$ of size $t=O(n^k)$. Constants of $\IC$ are replaced with elements
of a small finite field by considering the system:
$$\left\{
 \begin{array}{l}
   C_n(x, Y) = 0 \text{ for all } x \in L \cap \{0,1\}^n\\
   \left( \prod_{x \in \{0,1\}^n \setminus L} C_n(x, Y) \right) \cdot Z = 1.
 \end{array}
\right.$$
The end of the proof is similar.

For the third point, let us assume $\unifVNP \subset \asize_\IC(\poly)$.
It implies $\VP = \VNP$ thanks to the $\VNP$-completeness of the
uniform family $(\HC_n)$, then $\MA = \PP$ by Lemma~\ref{lem:collapse}.
This implies $\MA \not\subset \size(n^k)$ for all $k$
since $\PP \not\subset \size(n^k)$ for all $k$~\cite{Vinodchandran05}.

For the last point, assume $\NP=\MA$. If $\NP$ is without $n^k$ circuits
for all $k$, then the conclusion comes from the first point.
Otherwise $\MA$ has $n^k$-size circuits and the conclusion follows
from the previous point.
\end{proof}

For any constant $c$, the class $\P^{\NP[n^c]}$ is the set of languages
decided by a polynomial time machine making $O(n^c)$ calls to an $\NP$
oracle. It is proven in~\cite{FortnowSW09} that $\NP \subset \size(n^k)$
implies $\P^{\NP[n^c]} \subset \size(n^{ck^2})$. Hence, it is enough to assume
fixed-polynomial lower bounds on this larger class $\P^{\NP[n^c]}$ for some
$c$ to get fixed-polynomial lower bounds on $\unifVNP_\IC$.

\subsubsection*{An unconditional lower bound in characteristic zero}

In this part we do not allow arbitrary constants in circuits.
We consider instead circuits with $-1$ as the only scalar that can label the leaves. For $s : \IN \rightarrow \IN$, let $\asizez (s)$ be the family of polynomials computed by families of unbounded degree constant-free circuits of size $O(s)$ (in characteristic zero). Note that the formal degree of these circuits are not polynomially bounded: hence, large constants produced by small arithmetic circuits can be used.

We first need a result of~\cite{AllenderBKM09}. Let
$\mathrm{PosCoefSLP}$
be the following problem: on input $(C,i)$ where $C$ is a constant-free
circuit with one variable $x$ and $i$ is an integer, decide whether the
coefficient of $x^i$ in the polynomial computed by $C$ is positive.
\begin{lem}[\cite{AllenderBKM09}]\label{lem:poscoefslp}
  $\mathrm{PosCoefSLP}$ is in $\CH$.
\end{lem}
\begin{thm}
$\unifVNP \not\subset \asizez (n^k)$ for all $k$.
\end{thm}
\begin{proof}
  If the permanent family does not have constant-free arithmetic circuits of
  polynomial size, then this family matches the statement.

  Otherwise, $\CH=\MA$ by Lemma~\ref{lem:collapse}. For a given constant-free circuit
  $C$ computing a univariate polynomial $P=\sum_{i=0}^d a_ix^i$, its ``sign
  condition'' is defined as the series $(b_i)_{i\in\IN}$ where
  $b_i\in\{0,1\}$, $b_i=1$ iff $a_i>0$.

  Note that for some constant $\alpha$, there are at most $2^{n^{\alpha k}}$
  different sign conditions of constant-free circuits of size $n^k$ (at most
  one per circuit). Hence there exists a sign condition
$$(b_0, \dots, b_{n^{\alpha k}}, 0, 0, \dots)$$
such that any polynomial with such
  a sign condition is not computable by constant-free circuits of size
  $n^k$. We define $b_0,\dots,b_{n^{\alpha k}}$ to be the lexicographically
  first such bits.

  We can express these bits as the first in lexicographic order such that for
  every constant-free circuit $C$, there exists $i$ such that:
   $$\text{$b_i=0$ iff the
  coefficient of $x^i$ in $C$ is positive.}$$
 Therefore they can be computed in
  $\PH^{\mathrm{PosCoefSLP}}$, hence in $\CH$ by Lemma~\ref{lem:poscoefslp},
  hence in $\MA$ since $\CH=\MA$. By reducing the probability of error in the
  $\MA$ protocol, this means that there exists a polynomial-time function
  $a:\{0,1\}^\star\to\{0,1\}$ such that:
  $$\begin{cases}
    \exists y\sum_r a(i,y,r) \geq (1-2^{-|y|-1})N & \text{if }b_i=1\\
    \forall y\sum_r a(i,y,r) \leq 2^{-|y|-1}N & \text{if }b_i=0,
  \end{cases}$$
  where $y$ and $r$ are words of polynomial size, and where $N=2^{|r|}$. Now,
  the following polynomial family:
  $$P_n(x)=\sum_{i=0}^{n^{\alpha k}}\biggl(\bigl(\sum_{y,r}a(i,y,r)\bigr)-N/2\biggr)x^i$$
  is in $\unifVNP$ and has sign condition $(b_0,\dots,b_{n^{\alpha
      k}},0,0,\dots)$.
\end{proof}

\subsubsection*{In positive characteristic}

This subsection deals with fixed-polynomial lower bounds in positive characteristic. The results are presented in characteristic $2$ but they hold in any positive characteristic $p$ (replacing $\parityP$ with $\ModpP$).

\begin{lem}\label{lem:ham}
Consider the polynomial
$$P(X_1,\ldots,X_n) = \sum_{y_1,\ldots,y_p \in \{0,1\}}
C(X_1,\ldots,X_n,y_1,\ldots,y_p)$$
where $C$ is an arithmetic circuit of size $s$ and total degree at most $d$
(with respect to all the variables $X_1\ldots,X_n,y_1,\ldots,y_p$).
Then $P$ is a projection of $\HC_{(sd)^{O(1)}}$.
\end{lem}

\begin{proof}
This lemma follows from a careful inspection of the proof of $\VNP$-completeness of the Hamiltonian given in Malod~\cite{Malod03}. We give some more details below.

From the fact that $\VNP = \VNPe$~\cite[Theorem 2.13]{Burgisser00}, we can write $P$ as a Boolean sum of formulas, i.e.
$$P(X_1,\ldots,X_n) = \sum_{z_1,\ldots,z_q \in \{0,1\}}
F(X_1,\ldots,X_n,z_1,\ldots,z_q).$$
Moreover, $q=s^{O(1)}$ and an inspection of the proof of $\VNP = \VNP_e$ given in~\cite{Malod03} shows that the size of the formula $F$ is $(sd)^{O(1)}$.
By~\cite[Lemme 8]{Malod03}, a formula is a projection of the Hamiltonian circuit polynomial of linear size. This yields
$$P(X_1,\ldots,X_n) = \sum_{z_1,\ldots,z_q \in \{0,1\}} \HC_{s'}(a_1,\ldots,a_{s'})$$
where $s'=(sd)^{O(1)}$ and $a_i \in \{X_1,\ldots,X_n,z_1,\ldots,z_q,-1,0,1\}$.
At last, in order to write this exponential sum as a projection of a not too large Hamiltonian circuit, a sum gadget of size $O(q)$ and $O(s')$ XOR gadgets of size $O(1)$ are needed~\cite[Théorème 7]{Malod03}. Hence, the polynomial $P$ is a projection of $\HC_{(sd)^{O(1)}}$.
\end{proof}

\begin{thm}\label{thm:positive-char}
The following are equivalent:
\begin{itemize}
\item $\unifVNP_\IFtwo \subset \asize_\IFtwo (n^k)$ for some $k$;
\item $\VP_\IFtwo = \VNP_\IFtwo$ and
$\parityP \subset \size(n^k)$ for some $k$.
\end{itemize}
\end{thm}

\begin{proof}
Suppose that $\unifVNP_{\IFtwo} \subset \asize_\IFtwo (n^k)$. Then the  Hamiltonian polynomials $(\HC_n)$ has $O(n^k)$ size circuits and thus
$\VP = \VNP$ over $\IFtwo$. Let $L \in \parityP$ and the corresponding function
$f\in\sharpP$ so that
$$x \in L \iff f(x)\text{ is odd}.$$
Consider the sequence of polynomials $P_n \in \IFtwo[X_1,\ldots,X_n]$ associated to $L$:
$$P_n(X_1,\ldots,X_n) = \sumonwords{x}{n} f(x) \interpol{X}{x}{n}.$$
This family belongs to $\unifVNP$ over $\IFtwo$. Hence, $P_n$ has $O(n^k)$
size circuits. It can be simulated by a Boolean circuit of the same size
within a constant factor, and yields $O(n^k)$ size circuits for $L$. Hence
$\parityP \subset \size(n^k)$.

For the converse, suppose that $\parityP \subset \size(n^k)$ and $\VP_{\IFtwo} = \VNP_{\IFtwo}$, and let $(P_n) \in \unifVNP_{\IFtwo}$.
We can write
$$P_n (X_1,\ldots,X_n) = \sum_{m_1,\ldots,m_n \in \{0,\ldots,d\}} \phi(m_1,\ldots,m_n) \prod_{i=1}^n X_i^{m_i}$$
where $d$ is a bound on the degree of each variable of $P_n$.
Since the coefficients of $P_n$ belong to $\parityP$,
they can be computed by
Boolean circuits of size $O(\tilde{n}^k)$ with $\tilde{n} = n \log n$ (by our hypothesis on circuits
size for $\parityP$ languages and the fact that the function $\phi$ 
takes $n \log d$ bits).

These Boolean circuits can in turn be simulated by (Boolean) sums
of arithmetic circuits of size and formal degree $O(\tilde{n}^k)$
by the usual method (see e.g. the proof of Valiant's criterion in~\cite{Burgisser00}).

Hence we have written $P_n = \sum_{\tilde{m}} \psi(\tilde{m}) X^{\tilde{m}}$,
i.e. $P_n$ is a sum over $O(\tilde{n}^k)$ variables in $\IFtwo$ of an arithmetic circuit $\psi$ of size $O(\tilde{n}^k)$, and the degree of $\psi$ is $O(\tilde{n}^k)$.
By Lemma~\ref{lem:ham}, $P_n$ is a projection of $\HC_{\tilde{n}^{O(k)}}$.
By hypothesis, the uniform family
$(\HC_n)$ has $O(n^k)$ arithmetic circuits. Hence, $(P_n)$ has
arithmetic circuits of size 
$n^{O(k^2)}$.
\end{proof}

\section*{Acknowledgements}

We thank Guillaume Malod for useful discussions (in particular on
Lemma~\ref{lem:ham}) and Thomas Colcombet for some advice on the presentation.

\bibliographystyle{plain}
\bibliography{vnp}

\begin{thebibliography}{10}

\bibitem{AllenderBKM09}
Eric Allender, Peter B{\"u}rgisser, Johan Kjeldgaard-Pedersen, and Peter~Bro
  Miltersen.
\newblock On the complexity of numerical analysis.
\newblock {\em SIAM J. Comput.}, 38(5):1987--2006, 2009.

\bibitem{BalcazarDG90}
Jos{\'e}~Luis Balc{\'a}zar, Josep D{\'{\i}}az, and Joaquim Gabarr{\'o}.
\newblock {\em Structural complexity. {II}}, volume~22 of {\em EATCS Monographs
  on Theoretical Computer Science}.
\newblock Springer-Verlag, Berlin, 1990.

\bibitem{BaurS83}
Walter Baur and Volker Strassen.
\newblock The complexity of partial derivatives.
\newblock {\em Theor. Comput. Sci.}, 22:317--330, 1983.

\bibitem{Burgisser00}
Peter B{\"u}rgisser.
\newblock {\em Completeness and reduction in algebraic complexity theory},
  volume~7 of {\em Algorithms and Computation in Mathematics}.
\newblock Springer-Verlag, Berlin, 2000.

\bibitem{FortnowSW09}
Lance Fortnow, Rahul Santhanam, and Ryan Williams.
\newblock Fixed-polynomial size circuit bounds.
\newblock In {\em IEEE Conference on Computational Complexity}, pages 19--26,
  2009.

\bibitem{HemaspaandraO02}
Lane~A. Hemaspaandra and Mitsunori Ogihara.
\newblock {\em The complexity theory companion}.
\newblock Texts in Theoretical Computer Science. An EATCS Series.
  Springer-Verlag, Berlin, 2002.

\bibitem{HrubesY11}
Pavel Hrubes and Amir Yehudayoff.
\newblock Arithmetic complexity in ring extensions.
\newblock {\em Theory of Computing}, 7(1):119--129, 2011.

\bibitem{JansenS12}
Maurice~J. Jansen and Rahul Santhanam.
\newblock Stronger lower bounds and randomness-hardness trade-offs using
  associated algebraic complexity classes.
\newblock In {\em STACS}, pages 519--530, 2012.

\bibitem{KabanetsI04}
Valentine Kabanets and Russell Impagliazzo.
\newblock Derandomizing polynomial identity tests means proving circuit lower
  bounds.
\newblock {\em Computational Complexity}, 13(1-2):1--46, 2004.

\bibitem{Kannan82}
Ravi Kannan.
\newblock Circuit-size lower bounds and non-reducibility to sparse sets.
\newblock {\em Information and Control}, 55(1-3):40--56, 1982.

\bibitem{Koiran96}
Pascal Koiran.
\newblock Hilbert's {N}ullstellensatz is in the polynomial hierarchy.
\newblock {\em J. Complexity}, 12(4):273--286, 1996.

\bibitem{Koiran96rr}
Pascal Koiran.
\newblock Hilbert's {N}ullstellensatz is in the polynomial hierarchy.
\newblock Technical Report 96-27, DIMACS, July 1996.

\bibitem{Lipton75}
Richard~J. Lipton.
\newblock Polynomials with 0-1 coefficients that are hard to evaluate.
\newblock In {\em FOCS}, pages 6--10, 1975.

\bibitem{LundFKN90}
Carsten Lund, Lance Fortnow, Howard~J. Karloff, and Noam Nisan.
\newblock Algebraic methods for interactive proof systems.
\newblock In {\em FOCS}, pages 2--10, 1990.

\bibitem{Malod03}
Guillaume Malod.
\newblock {\em Polynômes et coefficients}.
\newblock PhD thesis, Université Claude Bernard Lyon 1, 2003.
\newblock http://tel.archives-ouvertes.fr/tel-00087399.

\bibitem{Raz10}
Ran Raz.
\newblock Elusive functions and lower bounds for arithmetic circuits.
\newblock {\em Theory of Computing}, 6(1):135--177, 2010.

\bibitem{Santhanam09}
Rahul Santhanam.
\newblock Circuit lower bounds for merlin--arthur classes.
\newblock {\em SIAM J. Comput.}, 39(3):1038--1061, 2009.

\bibitem{Schnorr78}
Claus-Peter Schnorr.
\newblock Improved lower bounds on the number of multiplications/divisions
  which are necessary of evaluate polynomials.
\newblock {\em Theor. Comput. Sci.}, 7:251--261, 1978.

\bibitem{Stockmeyer85}
Larry~J. Stockmeyer.
\newblock On approximation algorithms for \#p.
\newblock {\em SIAM J. Comput.}, 14(4):849--861, 1985.

\bibitem{Strassen74}
Volker Strassen.
\newblock Polynomials with rational coefficients which are hard to compute.
\newblock {\em SIAM J. Comput.}, 3(2):128--149, 1974.

\bibitem{Valiant79}
Leslie~G. Valiant.
\newblock Completeness classes in algebra.
\newblock In {\em STOC}, pages 249--261, 1979.

\bibitem{Vinodchandran05}
N.~V. Vinodchandran.
\newblock A note on the circuit complexity of {PP}.
\newblock {\em Theor. Comput. Sci.}, 347(1-2):415--418, 2005.

\end{thebibliography}

\end{document}